\documentclass[onecolumn,letter]{IEEEtran}

\usepackage{amssymb}
\usepackage{amsmath}
\usepackage{latexsym}
\usepackage{graphicx}
\usepackage{bm}

\newcommand{\defeq}{\stackrel{\rm def}{=}}

\def\complex{\mathbb{C}}

\def\Tr{\mathop{\rm Tr}\nolimits}

\newtheorem{thm}{Theorem}

\def\choose#1#2{\genfrac{(}{)}{0pt}{}{#1}{#2}}
\def\Label#1{\label{#1}\ [\ #1\ ]\ }
\def\Label{\label}

\begin{document}
\title{Capacity with energy constraint in coherent state channel}
\author{
Masahito Hayashi
\thanks{
M. Hayashi is with Graduate School of Information Sciences, Tohoku University, Aoba-ku, Sendai, 980-8579, Japan
(e-mail: hayashi@math.is.tohoku.ac.jp)}}
\date{}
\maketitle


\maketitle

\begin{abstract}
We consider two kind of energy constraints
when the output state is a coherent state.
One is a constraint on the total energy during a fixed period;
the other is a constraint on the total energy for a single code.
The first setting can be easily dealt with by using the conventional capacity formula.
The second setting requires the general capacity formula for a classical-quantum channel.
\end{abstract}

\begin{keywords}
Energy constraint,
Coherent state,
Information spectrum,
average error probability,
pulse
\end{keywords}

\section{Introduction}
Recently, 
the demand for high speed optical communication 
has been steadily increasing.
Optical communication is a communication technology sending information 
via transmitting photons through optical fibers.
Optical communication requires energy
because photons have energy, which increases in proportion to the average photon number.
For reasons of economy, the average photon number is restricted to a fixed value.

In conventional information theory, including quantum information theory, it has been usual to deal with the capacity, which is the maximum transmission rate with an energy constraint.
In this formulation, the average photon number per single pulse is restricted.
That is, the total average photon number is allowed to increase in proportion to the number of pulses in a single code.
However, in real optical communication,
the total average photon number over a fixed period, rather than the average photon number per pulse,
should be restricted to a fixed value.

Since coherent light is normally used for optical communication, 
it is natural to assume that only coherent states are available for the signal states.
The attenuation channel is often assumed as the quantum communication channel.
In this case, the output state of the coherent input state
is also a coherent state.
Hence, we only treat the case where the output state is a coherent state.

In the present paper, with respect to the first setting,
the amount of transmitted information 
when the total average photon number $E$ during the fixed period
is fixed and the number $K$ of pulses during the fixed period is increasing can be changed.

As for the second setting, 
the amount of transmitted information is treated
when the total average photon number $E$ for a single code is fixed 
and the number $N$ of pulses for a single code is increasing.
In this framework, the relation between the 
amount of transmitted information and the average error probability
in the asymptotic setting is discussed.
Since the restriction for a single pulse depends on the number $N$ of pulses,
we cannot apply the conventional capacity formula for the stationary memoryless channel.
In order to resolve this problem, 
we apply the asymptotic general capacity formula for a classical-quantum channel,
which was invented by Hayashi-Nagaoka\cite{H-N} as the quantum version of 
Verd\'{u}-Han\cite{Verdu-Han}'s general capacity formula.
In this formula, the quantum information spectrum plays an essential role.

In the third formulation,
we evaluate the average error probability
based only on 
the average photon number $C$ of coherent light 
and the amount $N$ of transmission information
in the non-asymptotic setting.
In this discussion, Holevo's covariant measurement plays an essential role.
We derive a general relation between the average photon number and the average error probability
in a general framework.
Some results in the second setting
are recovered from this non-asymptotic formula.
Therefore, we discuss the required average photon number for reliable communication based on coherent light from several viewpoints.

The remainder of the present paper is set out as follows.
In section \ref{s2}, we deal with the first setting, i.e., 
the amount of transmitted information 
when the total average photon number $C$ during the fixed period
is fixed and the number $K$ of pulses during the fixed period is increasing.
In section \ref{s3}, we consider the second and third settings, i.e.,
Theorem \ref{th-9} is presented (in the second setting) and Theorem \ref{th6} (in the third setting) is presented as the main results.
In section \ref{s4}, we revisit the general capacity formula for a classical-quantum channel for the second setting
and derive a useful general formula for the case when the output states are pure.
This general formula is proven in the Appendix.
In section \ref{s5}, 
a proof of Theorem \ref{th-9} is given based on the discussion in section \ref{s4}.
In section \ref{s6}, 
a proof of Theorem \ref{th6} is given based on group representation theory.
In the appendix, an important theorem stated in section \ref{s4} is proved.

\section{Photon number constraint for a fixed period}\Label{s2}
In this paper, we treat the $N$-fold tensor product system 
${\cal H}^{\otimes N}$ of the Boson-Fock space ${\cal H}$, which is spanned by the number states $|0 \rangle, |1 \rangle, \ldots, |n \rangle, \ldots$.
Coherent light with the complex amplitude $\alpha$
is given as $|\alpha\rangle :=
e^{-\frac{|\alpha|^2}{2}}\sum_{n=0}^{\infty}
\frac{\alpha^n}{\sqrt{n!}} |n \rangle$.

In coding theory, the classical input symbols are called input alphabets,
and in our setting, they equal the complex plane $\complex$.
Then, the classical-quantum channel discussed here is given as a map from $\complex$ 
to the set of density operators on ${\cal H}$ of the form $\alpha \mapsto 
|\alpha\rangle\langle \alpha|$.
The $N$-fold memoryless extension is given as
a map from $\complex^N$ to the set of density matrices on the $N$-th tensor product system ${\cal H}^{\otimes N}$. 
That is, this extension maps the input sequence $\vec{\alpha}=(\alpha_1, \ldots, \alpha_N)$ 
to the state 
$|\vec{\alpha}\rangle\langle \vec{\alpha}|
:=|\alpha_1, \ldots, \alpha_N\rangle \langle\alpha_1, \ldots, \alpha_N|$.

Sending the message $\{1, \ldots, M_N\}$ requires an encoder and a decoder.
The encoder is given as a map $\varphi_N$ from the set of messages $\{1, \ldots, M_N\}$ to the set of alphabets $\complex^N$, and the decoder is given by a POVM $Y^N=\{Y_i^N\}_{i=1}^{M_N}$.
The triplet $\Phi_N:=(M_N, \varphi_N, Y^N)$ is called a code.
Its performance is evaluated by the value $|\Phi_N|:=M_N$ and
the average error probability, given by 
\begin{align*}
\varepsilon [\Phi_N]:=\frac{1}{M_N} \sum_{i=1}^{M_N} 
\langle \varphi_N(i)| (I-Y_i^N)|\varphi_N(i)\rangle.
\end{align*}
The code $\Phi_N:=(M_N, \varphi_N, Y^N)$ is required to satisfy the 
average photon number constraint
\begin{align}
\langle \varphi_N(i)|\hat{N}|\varphi_N(i)\rangle \le N E
\hbox{ for }\forall i
\Label{ene}
\end{align}
where the number operator $\hat{N}$ is given as $\sum_{n_1,\ldots, n_N}
(n_1+\ldots+n_N)| n_1,\ldots, n_N\rangle \langle n_1,\ldots, n_N|$.
This condition is equivalent to the condition
\begin{align*}
\|\vec{\alpha}\|^2
=\sum_{i=0}^{N}|\alpha_i|^2
\le N E
\end{align*}
where $\vec{\alpha}=\varphi_N(i)$.

Then, the channel capacity with the average photon number constraint (\ref{ene})
is given by
\begin{align*}
C(E)\defeq 
\sup_{\{\Phi_N\}_{N=1}^{\infty}}
\left\{
\liminf_{N \to \infty}
\frac{1}{N}\log |\Phi_N| 
\left|
\begin{array}{l}
\varepsilon [\Phi_N] \to 0 \\
\hbox{(\ref{ene}) holds}
\end{array}
\right. \right\}.
\end{align*}
The capacity $C(E)$ is calculated as\cite{GGL}
\begin{align*}
C(E)=\max_{\Tr \rho \hat{N} = E}H(\rho)=
(E+1)\log(E+1) -E \log E,
\end{align*}
where $H(\rho)=-\Tr \rho \log \rho$.

Hence, when the number of pluses during a fixed period is restricted to $K$
and the average photon number per second is limited by $E$, 
the bound of the possible amount of transmitted information is given as
\begin{align*}
K C(\frac{E}{K})=K ((\frac{E}{K}+1)\log(\frac{E}{K}+1) -\frac{E}{K} \log \frac{E}{K}). 
\end{align*}
When the condition $K$ is fixed, this limit is a finite value.
However, if the number $K$ is sufficiently large,
this value has the following asymptotic expansion:
\begin{align*}
K C(\frac{E}{K})
\cong E \log K +E - E\log E +\frac{E^2}{2} \frac{1}{K},
\end{align*}
which implies that
we can increase the number of sending bits by increasing the number $K$ of sending pulses per second
with the fixed average photon number constraint.

This type of phenomenon does not occur in the classical Gaussian channel.
When the channel noise is specified as a Gaussian distribution with variance $V$,
and the average photon number constraint is given as the condition that the input intensity per signal is $E$,
then the capacity is
\begin{align*}
C_c(E,V):=\frac{1}{2}\log(1+\frac{E}{V}).
\end{align*}
Hence, 
\begin{align*}
K C_c(\frac{E}{K},V)&=K \frac{1}{2}\log(1+\frac{E}{KV})
\cong \frac{E}{2V} -\frac{E^2}{4V^2}\frac{1}{K}\\
K C_c(\frac{E}{K},V) &\le \frac{E}{2V}.
\end{align*}
Thus,
even though the number $K$ of pulses per second increases,
the capacity is limited to $\frac{E}{2V}$.

One might consider that this comparison is inappropriate because 
the coherent state case is noiseless and only has attenuation whereas
the classical Gaussian case does include noise.
However, 
the variance in the estimation of the state family 
$\{|\alpha \rangle \langle \alpha| |\alpha \in \complex\}$
behaves as $O(\frac{1}{N})$ asymptotically when $N$ copies of the unknown state are prepared.
This behavior coincides with 
the estimation of the unknown expectation parameter
of the classical Gaussian distribution with fixed variance.
The former uncertainty is caused by quantum non-commutativity, and 
the latter uncertainty is caused by classical noise.
Usually, both cases have a similar asymptotic behavior, as happens with statistical state inference.
However, for the situation here, the cases have different 
asymptotic behaviors.
Thus, the difference discussed here can be regarded as a special phenomenon 
in the case of the capacity of the photon number constraint for a fixed period.

\section{Photon number constraint for a single code}\Label{s3}
In this section, we consider the relation between the transmission amount and the average photon number from another viewpoint.
First, we consider the transmission amount
when the total average photon number is restricted to the value $E$.
Our condition (\ref{ene}) is replaced by 
\begin{align}
\langle \varphi_N(i)|\hat{N}|\varphi_N(i)\rangle \le E
\hbox{ for }\forall i.
\Label{ene2}
\end{align}
Since the number of transmitted bits is of order $\log N$,
we define the capacity with respect to the logarithmic order
with error probability $\epsilon$ as
\begin{align*}
C_l(\epsilon,E)\defeq 
\sup_{\{\Phi_N\}_{N=1}^{\infty}}
\left\{
\liminf_{N \to \infty}
\frac{\log |\Phi_N| }{\log N}
\left|
\begin{array}{l}
\lim_{N\to \infty} \varepsilon [\Phi_N] \le \epsilon \\
\hbox{(\ref{ene2}) holds}
\end{array}
\right. \right\}.
\end{align*}

\begin{thm}\Label{th-9}
The capacity $C_l(\epsilon,E)$
is as follows:
\begin{align}
C_l(\epsilon,E)= 
\sup_m\left\{m\left| \sum_{n=0}^m e^{-E}\frac{E^n}{n!}\le \epsilon
\right.\right\}\Label{1-30-13}
\end{align}
\end{thm}
To prove this result, we cannot apply the conventional formula for a memoryless channel.
So, we prepare an information spectrum method in Section \ref{s4},
and present a proof in Section \ref{s5}.

In the above theorem,
we consider the transmission rate to logarithmic order 
with the condition for the average photon number and number of pulses
in the asymptotic setting.
In the following, we consider the 
transmission size
under the average photon number constraint
without any constraint on the number of pulses in either the non-asymptotic setting
or the asymptotic setting.
In this framework, we can use any number of pulses for information transmission, 
assuming the average photon number constraint.
The following theorem holds.
\begin{thm}\Label{th6}
When any state $|\varphi(i)\rangle$ of a code $\Phi$
satisfies (\ref{ene2}),
the inequality
\begin{align}
\varepsilon [\Phi]
\ge
1-\left(
\frac{1}{|\Phi|}\sqrt{1+(|\Phi|-1)e^{-E}}
+(1-\frac{1}{|\Phi|})\sqrt{1-e^{-E}}
\right)^2 \Label{2-1-7}
\end{align}
holds.
\end{thm}
Now, we denote the right hand side of (\ref{2-1-7}) when $|\Phi_N|=e^{R}$ by $\varepsilon(E,R)$.
When $R$ goes to infinity, 
we obtain its asymptotic expansion, which
depends on the behavior of $E-R$:
\begin{align*}
\varepsilon(E,R)
\cong
\left\{
\begin{array}{ll}
e^{-E}+ 2 \sqrt{1-e^{-E}}e^{-(E+R)/2} - (1-2e^{-E})e^{-R}
& \hbox{ if } E-R \to -\infty\\
\left(1+2 e^{A}-2\sqrt{e^{A}(1+e^{A})}\right) e^{-E} < e^{-E}
& \hbox{ if } E-R\to A\\
\frac{1}{4} e^{-2E+R} << e^{-E}
& \hbox{ if } E-R \to \infty.
\end{array}
\right.
\end{align*}
When the average photon number $E$ is fixed and the information size $R$ is sufficiently large,
the error probability is greater than $e^{-E}$.
Thus, the reliable transmission of a large amount of information requires a high average photon number.
This evaluation does not depend on the number of pulses.
In other words, in order to keep the error probability smaller than $p$, we need the average photon number to be at least $-\log p$
.
This observation coincides with that of Theorem \ref{th-9}.
Conversely, when the average photon number $E$ is large relative to the information transmission size $R$,
the obtained lower bound of the error probability rapidly approaches zero.
That is, this lower bound does not yield the crucial bound for the error probability.
In this case, in order to evaluate the error probability more precisely,
we need to take account of the number of pulses.

\section{Information Spectrum Approach to a Pure States Channel}\Label{s4}
For a proof of Theorem \ref{th-9},
we cannot apply the conventional capacity formula for the stationary memoryless channel.
Instead we employ the asymptotic general capacity formula for general sequences of a classical-quantum channel by Hayashi-Nagaoka\cite{H-N}.
In the general capacity formula,
we focus on the sequence of Hilbert spaces $\{{\cal H}^{(N)}\}$, sets of alphabets 
$\{{\cal X}^{(N)}\}$, and 
classical-quantum channels $\bm{W}:=\{W^{(N)}\}$, where
the channel $W^{(N)}$ maps an alphabet $x \in {\cal X}^{(N)}$ 
to the density operator $W^{(N)}_x$ on ${\cal H}^{(N)}$.
For any sequence $\{a_N\}$ satisfying the condition $a_N \to \infty$,
we define the capacity
\begin{align*}
C(\epsilon|\bm{W})\defeq 
\sup_{\{\Phi_N\}_{N=1}^{\infty}}
\left\{
\liminf_{N \to \infty}
\frac{\log |\Phi_N| }{a_N}
\left|
\lim_{N\to \infty} \varepsilon [\Phi_N] \le \epsilon 
 \right. \right\}, 
\end{align*}
where $\Phi_N$ expresses a code for the classical-quantum channel $W^{(N)}$.
For any sequence of probability distributions $P^{(N)}$ on ${\cal X}^{(N)}$,
we define the information spectrum quantity
\begin{align*}
I(\epsilon| \bm{P},\bm{W})
:=
\sup\left\{b\left|
\limsup_{N \to \infty}
\int_{{\cal X}^{(N)}}
\Tr W_x^{(N)} 
\{W_x^{(N)} -e^{a_N b} W_{P^{(N)}}^{(N)} \le 0 \}
P^{(N)}(dx) \le \epsilon
\right. \right\},
\end{align*}
where 
the projection $\{X\ge 0 \}$ is defined as $\sum_{i:x_i\ge0}E_i$
when the spectral decomposition of $X$ is given as $\sum_{i} x_i E_i$, and
\begin{align*}
W_{P}^{(N)}
:=
\int_{{\cal X}^{(N)}}
W_x^{(N)} P(dx).
\end{align*}
Then, we obtain the following formula:
\begin{align}
C(\epsilon|\bm{W})
=
\sup_{\bm{P}}I(\epsilon| \bm{P},\bm{W}), \Label{1-27-3}
\end{align}
where $\bm{P}$ is a sequence $\{P^{(N)}\}$ of distributions.
This formula can be obtained by combining
Theorem 6 of Verd\'{u}-Han \cite{Verdu-Han} and 
Lemmas 3 and 4 of Hayashi-Nagaoka\cite{H-N}, as is mentioned in Remarks 8 and 11 in Hayashi-Nagaoka\cite{H-N}.

When $W_x^{(N)}$ is a pure state for all $x \in {\cal X}^{(N)}$, 
the quantity $I(\epsilon| \bm{P},\bm{W})$ can be characterized by the information spectrum quantity, $\bm{W_P}:= \{W_{P}^{(N)} \}$.
For a sequence of density operators 
$\bm{\rho}:=\{\rho^{(N)}\}$, we define 
$\overline{H}_+(\epsilon| \bm{\rho})$ by
\begin{align*}
\overline{H}_+(\epsilon| \bm{\rho})
:= \sup_{b}
\left\{b \left|
\limsup_{N \to\infty}
\Tr \rho^{(N)}
\left\{
\frac{-1}{a_N}\log \rho^{(N)}\le b
\right\}
\le \epsilon
\right.\right\}.
\end{align*}
Then, when $W_x^{(N)}$ is a pure state for all $x \in {\cal X}^{(N)}$, 
the following theorem holds.
\begin{thm}\Label{th1}
Assume that $W_x^{(N)}$ is a pure state for all $x \in {\cal X}^{(N)}$.
Then, the relation
\begin{align}
\overline{H}_+(\epsilon| \bm{W_P})
=
I(\epsilon| \bm{P},\bm{W})\Label{1-30}
\end{align}
holds for any real number $0 \le \epsilon < 1$.
\end{thm}
The proof of this result will be given in the Appendix.

Combining (\ref{1-27-3}) and (\ref{1-30}), 
we obtain
\begin{align}
C(\epsilon|\bm{W})
=
\sup_{\bm{P}}
\overline{H}_+(\epsilon| \bm{W_P}). 
\Label{1-30-2}
\end{align}

Further, the quantity 
$\overline{H}_+(\epsilon| \bm{\rho})$
satisfies the convexity requirement as follows.
\begin{thm}\Label{th2}
For any two sequences 
$\bm{\rho_i}:=\{\rho^{(N)}_i\}$ ($i=1,2$) 
and $0 \le t \le 1$,
we define the sequences 
$\bm{\rho_3}:=\{
t \rho^{(N)}_1+(1-t) \rho^{(N)}_2\}$.
When 
$\rho^{(N)}_1$ is unitarily equivalent to 
$\rho^{(N)}_2$, then
\begin{align}
\overline{H}_+(\epsilon| \bm{\rho_3})
\ge
\overline{H}_+(\epsilon| \bm{\rho_1}). \Label{1-30-10}
\end{align}
\end{thm}
\begin{proof}
In order to apply the majorization theory for eigenvalues,
we denote the $j$-th eigenvalue of the matrix $X$ by $\lambda_j(X)$, listing the eigenvalues in decreasing order.
Using Theorem III.4.1 of Bhatia \cite{Bha},
we obtain 
\begin{align*}
\sum_{j=1}^{k}\lambda_j( \rho_3^{(N)})
\le
\sum_{j=1}^{k}\lambda_j( t \rho_1^{(N)})
+
\sum_{j=1}^{k}\lambda_j((1- t) \rho_2^{(N)})
=
\sum_{j=1}^{k}\lambda_j( \rho_1^{(N)}),
\end{align*}
which implies (\ref{1-30-10}).
\end{proof}

\section{Application of the information spectrum approach}\Label{s5}
We apply the formula (\ref{1-30-2}) to our problem.
Then, we obtain
\begin{align*}
C_l(\epsilon,E)= 
\sup_{\bm{P}}
\overline{H}_+(\epsilon| \bm{\sigma_P}),
\end{align*}
where 
$a_N=\log N$,
$\sigma_{P^{(N)}}^{(N)}:= 
\int_{\complex^N}|\vec{\alpha}\rangle \langle \vec{\alpha}|
P^{(N)}(d \vec{\alpha})$,
and 
the support of $P^{(n)}$ is $\{
\vec{\alpha}\in \complex^N| \|\vec{\alpha}\|^2 \le E\}$.
Theorem \ref{th2} guarantees that
\begin{align}
C_l(\epsilon,E)= 
\sup_{\bm{P}\in \bm{{\cal P}}_{inv}}
\overline{H}_+(\epsilon| \bm{\sigma_P}),\Label{1-30-21}
\end{align}
where the sequence of distributions 
$\bm{{\cal P}}_{inv}:=\{ {\cal P}_{inv}^{(N)}\}$
is defined as
\begin{align*}
{\cal P}_{inv}^{(N)}:=
\{P^{(N)}|
P^{(N)} \hbox{ and this sequence is invariant under any action of }
U(\complex^N)
\}.
\end{align*}
For any distribution $P^{(N)}\in {\cal P}_{inv}^{(N)}$,
there exists a distribution $\bar{P}^{(N)}$ on $[0,\sqrt{E}]$ 
such that
\begin{align*}
P^{(N)}(d \vec{\alpha})=
\bar{P}^{(N)}(dr )\mu_N (d \Omega),
\end{align*}
where $\vec{\alpha}=r \Omega$
and $\mu_N$ is the invariant measure on the $N-1$-dimensional sphere.
In the following, we prove that
\begin{align}
\liminf_{N \to\infty}
\Tr \sigma_{P^{(N)}}^{(N)}
\{
\frac{-1}{\log N}\log \sigma_{P^{(N)}}^{(N)} \le c
\}
\ge
\sum_{n=0}^{m}
e^{-E}
\frac{E^{n}}{n!} \Label{2-1-5}
\end{align}
when $c$ is not an integer and $m$ is the maximum integer less than $c$.
The equality holds when the distribution $\bar{P}^{(N)}$ is the delta measure on 
$\{\sqrt{E}\}$.
Thus, 
$\sup_{\bm{P}\in \bm{{\cal P}}_{inv}}
\overline{H}_+(\epsilon| \bm{\sigma_P}) =
\sup_m\{m| \sum_{n=0}^m e^{-E}\frac{E^n}{n!}\le \epsilon\}
$.
Combining with (\ref{1-30-21}), we obtain (\ref{1-30-13}).

Next, we prove (\ref{2-1-5}). 
The state $\sigma_{P^{(N)}}^{(N)}$ can be written as
\begin{align*}
\sigma_{P^{(N)}}^{(N)}
&=
\int_{[0,\sqrt{E}]}
\sum_{n=0}^{\infty}
e^{-r^2}
\frac{r^{2n}}{n!}
\frac{1}{\choose {N+n-1}{N-1}}
\Pi_{n,N}
\bar{P}^{(N)}(dr )
=\sum_{n=0}^{\infty}
\lambda_n^N
\Pi_{n,N},
\end{align*}
where the projection $\Pi_{n,N}$ and the eigenvalue $\lambda_n^N$
are defined as
$\Pi_{n,N}:=
\sum_{\vec{n},\|\vec{n}\|=n}
|\vec{n}\rangle \langle\vec{n}|$
and
$\lambda_n^N:=
(\int_{[0,\sqrt{E}]}
e^{-r^2}
\frac{r^{2n}}{n!}
\bar{P}^{(N)}(dr ))
\frac{1}{\choose {N+n-1}{N-1}}$.

Since the eigenvalue is evaluated by
\begin{align*}
\frac{1}{N^{n}}
\ge
(\int_{[0,\sqrt{E}]}
e^{-r^2}
\frac{r^{2n}}{n!}
\bar{P}^{(N)}(dr ))
\frac{1}{N^{n}}
\ge
(\int_{[0,\sqrt{E}]}
e^{-r^2}
\frac{r^{2n}}{n!}
\bar{P}^{(N)}(dr ))
\frac{1}{\choose {N+n-1}{N-1}},
\end{align*}
$\{
\frac{-1}{\log N}\log \sigma_{P^{(N)}}^{(N)} \le c
\}
\le
\sum_{n=0}^m 
\Pi_{n,N}$.
Hence, the probability
$\Tr \sigma_{P^{(N)}}^{(N)}
\{
\frac{-1}{\log N}\log \sigma_{P^{(N)}}^{(N)} \le c
\}
$ is evaluated by
\begin{align}
\Tr \sigma_{P^{(N)}}^{(N)}
\{
\frac{-1}{\log N}\log \sigma_{P^{(N)}}^{(N)} \le c
\}
\le
\sum_{n=0}^{m}
\int_{[0,\sqrt{E}]}
e^{-r^2}
\frac{r^{2n}}{n!}
\bar{P}^{(N)}(dr ),\Label{2-1-1}
\end{align}
where the integer $m$ is the maximum integer less than $c$.
Now, we treat the opposite inequality when $c$ is not an integer
and 
\begin{align}
N \ge e^{E/c}\Label{2-6}
\end{align}
by considering $m+1$ cases: namely the cases $(0),(1), \ldots, (m)$.

Case (0):
Assume that the inequality
\begin{align}
(\int_{[0,\sqrt{E}]}
e^{-r^2}
\frac{r^{2n}}{n!}
\bar{P}^{(N)}(dr ))
\frac{1}{\choose {N+n-1}{N-1}}
\ge N^{n-c} \Label{2-1-2}
\end{align}
holds for all $n\le m$. 
Then, for $n \le m$,
\begin{align*}
\frac{1}{N^{c}}
\le
(\int_{[0,\sqrt{E}]}
e^{-r^2}
\frac{r^{2n}}{n!}
\bar{P}^{(N)}(dr ))
\frac{1}{\choose {N+n-1}{N-1}},
\end{align*}
which implies that
$\frac{-1}{\log N} \log \lambda_n^N \le c$.
Thus, the equality of (\ref{2-1-1}) holds.

Case ($n$) ($n=1, \ldots, m$):
Assume that the inequality (\ref{2-1-2}) does not hold for the integer $n$.
Note that the inequality (\ref{2-1-2}) always holds for $n=0$ for sufficiently large $N$.
Since $x\mapsto x^n$ is a convex function,
the average $r^2_N:= 
\int_{[0,\sqrt{E}]}
r^2 \bar{P}^{(N)}(dr )$
satisfies 
\begin{align*}
e^{-E}\frac{(r^2_N)^n}{n!}
\frac{1}{(1+\frac{n-1}{N})^{n}}
\le 
e^{-E}\frac{
\int_{[0,\sqrt{E}]}
r^{2n}
\bar{P}^{(N)}(dr )
}{n!}
\frac{1}{(1+\frac{n-1}{N})^{n}}
\le 
(\int_{[0,\sqrt{E}]}
e^{-r^2}
\frac{r^{2n}}{n!}
\bar{P}^{(N)}(dr ))
\frac{N^{n}}{\choose {N+n-1}{N-1}}
< N^{n-c} 
\end{align*}
because 
$N^n 
(1+\frac{n-1}{N})^{n}
\ge 
\choose {N+n-1}{N-1}$.
That is,
\begin{align*}
r^2_N
< L_n(N):=
\left(
e^{E}n!
(1+\frac{n-1}{N})^{n}
N^{n-c} 
\right)^{\frac{1}{n}}.
\end{align*}
The eigenvalue corresponding to the vector $|0,\ldots, 0\rangle$
is
$
\int_{[0,\sqrt{E}]}
e^{-r^2}
\bar{P}^{(N)}(dr )$, which is larger than 
$e^{-r_N^2}$ because $x \mapsto e^{-x}$ is convex.
Since 
the condition (\ref{2-6}) guarantees that
$e^{-r_N^2} \ge e^{-E} \ge \frac{1}{N^{c}}$,
we obtain
$\frac{-1}{\log N} \log \lambda_0^N \le c$.
Thus,
\begin{align*}
\Tr \sigma_{P^{(N)}}^{(N)}
\{
\frac{-1}{\log N}\log \sigma_{P^{(N)}}^{(N)} \le c
\}
\ge
e^{-L_n(N)}.
\end{align*}
Since the right hand side of (\ref{2-1-1}) is less than $1$,
\begin{align*}
\sum_{n=0}^{m}
\int_{[0,\sqrt{E}]}
e^{-r^2}
\frac{r^{2n}}{n!}
\bar{P}^{(N)}(dr )
-
\Tr \sigma_{P^{(N)}}^{(N)}
\{
\frac{-1}{\log N}\log \sigma_{P^{(N)}}^{(N)} \le c
\}
\le
1- e^{-L_n(N)}.
\end{align*}
Considering all cases $(0),(1), \ldots, (m)$,
we obtain
\begin{align*}
\sum_{n=0}^{m}
\int_{[0,\sqrt{E}]}
e^{-r^2}
\frac{r^{2n}}{n!}
\bar{P}^{(N)}(dr )
-
\Tr \sigma_{P^{(N)}}^{(N)}
\{
\frac{-1}{\log N}\log \sigma_{P^{(N)}}^{(N)} \le c
\}
\le
\max_{1 \le n \le m}
1- e^{-L_n(N)},
\end{align*}
which goes to $0$.
Therefore,
\begin{align*}
\liminf_{N \to\infty}
\Tr \sigma_{P^{(N)}}^{(N)}
\{
\frac{-1}{\log N}\log \sigma_{P^{(N)}}^{(N)} \le c
\}
=
\liminf_{N \to\infty}
\int_{[0,\sqrt{E}]}
\sum_{n=0}^{m}
e^{-r^2}
\frac{r^{2n}}{n!}
\bar{P}^{(N)}(dr )
\ge
\sum_{n=0}^{m}
e^{-E}
\frac{E^{n}}{n!},
\end{align*}
where equality holds when $\bar{P}^{(N)}$ is the delta measure on $\{\sqrt{E}\}$.
Therefore, we obtain (\ref{2-1-5}).

\section{Group Covariant Approach}\label{s6}
First, we consider the case when 
the state $|\varphi(i)\rangle$ is given as
$|f_i\rangle:=
\sqrt{p} |0\rangle + \sqrt{1-p}|i\rangle $
where $p=e^{-E}$,
$M=|\Phi|$ and $|0\rangle, |1\rangle, \ldots,|M \rangle$
are orthogonal to each other.
We focus on the permutation group $S_M$ whose representation $V$ 
is given as
\begin{align*}
V_g(|i\rangle)=
|g(i)\rangle, \quad
V_g(|0\rangle)=
|0\rangle
\end{align*}
for any $g \in S_M$.
Any one-dimensional subspace of ${\cal K}_1:=< |0\rangle, |S\rangle >$
is an irreducible space, where $|S\rangle
:= \frac{1}{\sqrt{M}}\sum_{j=0}^M |j \rangle$.
The remaining irreducible space ${\cal K}_2$ is the orthogonal space of ${\cal K}_1$.
We now define the states 
$|S'\rangle$ and $|i'\rangle$
as
\begin{align*}
|S'\rangle
&:=
\frac{1}{\sqrt{p+(1-p)/M}}
(\sqrt{p} |0\rangle+ \sqrt{1-p} \frac{1}{\sqrt{M}}|S\rangle)\\
|i'\rangle
&:=
\frac{\sqrt{M}}{\sqrt{M-1}}(|i\rangle - \frac{1}{\sqrt{M}}|S\rangle).
\end{align*}
Note that $|S'\rangle$ belongs to ${\cal K}_1$,
and $|i'\rangle$ belongs to ${\cal K}_2$.
Then, the state $|f_i\rangle$ can be written as
\begin{align*}
|f_i\rangle=
\sqrt{p+(1-p)/M} |S'\rangle +
\sqrt{(1-p)\frac{M}{M-1}}|i'\rangle.
\end{align*}
Thus, all states $|f_i\rangle$ belong to the space 
${\cal K}_3 \oplus {\cal K}_2$, where ${\cal K}_3$ is the one-dimensional space spanned by 
$|S'\rangle$.
Since the average correct probability 
$\frac{1}{M}
\sum_{j=1}^M \langle f_i |Y_i| f_i\rangle$ is invariant with respect to the action of the permutation group:
\begin{align*}
\frac{1}{M}
\sum_{j=1}^M \langle f_i |Y_i| f_i\rangle
=
\frac{1}{M}
\sum_{j=1}^M \langle f_i |
V_g
Y_{g^{-1}(i)}
V_g^{\dagger}
| f_i\rangle, \quad \forall g \in S_M,
\end{align*}
we can apply Holevo\cite{Hol}'s group covariant measurement theory.
Thus, our optimization problem can be 
restricted to an optimization problem among the POVM 
$\{Y_j\}_j$ of the following form:
\begin{align*}
Y_j= | u_j \rangle \langle u_j|, 
\end{align*}
where $| u_j \rangle
=\frac{1}{\sqrt{M}}|S'\rangle
+\frac{\sqrt{M-1}}{\sqrt{M}}| v_j \rangle$ 
and $| v_j \rangle$ is a unit vector in ${\cal K}_2$.
For this restriction,
the maximum value of $\langle f_i |Y_i| f_i\rangle$
is realized when $| v_j \rangle= |i'\rangle$,
and the maximum value is 
\begin{align*}
\left(
\frac{1}{\sqrt{M}} \sqrt{p+(1-p)/M} 
+
\sqrt{(1-p)\frac{M}{M-1}} \frac{\sqrt{M-1}}{\sqrt{M}}
\right)^2,
\end{align*}
which is equal to 
$
\left(
\frac{1}{M}\sqrt{1+(M-1)p}
+(1-\frac{1}{M})\sqrt{1-p}
\right)^2$.
Therefore, we obtain (\ref{2-1-7}).

Next, we prove the general case.
The state $|\varphi(i)\rangle$ has the form
$\sqrt{p} |e_0\rangle + \sqrt{1-p}|e_i\rangle $.
However, 
the vectors $|e_0\rangle, |e_1\rangle, \ldots, |e_M \rangle$
are not necessarily orthogonal to each other.
Define the TP-CP map ${\cal E}$ from the system spanned by $\{|f_1\rangle, \ldots,|f_M \rangle\}$ to the system spanned by $\{|e_0\rangle, |e_1\rangle, \ldots,|e_M \rangle\}$ as follows.
Here, $|f_j\rangle$ is the vector defined above.
First, we operate with the unitary operator $U$:
$U|f_j\rangle= |\varphi(i)\rangle\otimes |h_j\rangle$,
where the states $|h_j\rangle$ are an orthogonal basis on the additional system ${\cal K}_4$.
Next, we execute a partial trace with respect to the additional system ${\cal K}_4$.
Then, the TP-CP map ${\cal E}$ is defined as
\begin{align*}
{\cal E}(\rho):=\Tr_{{\cal K}_4}U \rho U^{\dagger}.
\end{align*}
Thus, any POVM $\{Y_j\}$ satisfies 
\begin{align*}
\langle \varphi(j)| Y_j |\varphi(j)\rangle
=
\langle f_j| {\cal E}^{\dagger} (Y_j )| f_j \rangle.
\end{align*}
Since $\{{\cal E}^{\dagger} (Y_j )\}$ satisfies the condition for POVM on 
the system spanned by $\{|f_1\rangle, \ldots,|f_M \rangle\}$,
we obtain the inequality (\ref{2-1-7}).

\section{Discussion}
The present paper discusses the relations between 
the average photon number constraint, the amount of transmitted information, 
and the average error probability.
While the second setting is based on an asymptotic framework,
the analysis of mutual information is not sufficient in this setting and an information spectrum approach is required.

As is shown in the second and third settings,
the average error probability is greater than $e^{-E}$ 
when the total average photon number is $E$ and the amount of transmitted information is sufficiently large.
In particular, Theorem \ref{th-9} guarantees that the minimum error probability $e^{-E}$ is realized when
the amount of the transmitted information is the logarithm of the number of transmitted pulses.
Indeed, such a code can be constructed as follows.
For an arbitrary integer $N$,
we define a code $\Phi$ with the set of messages
$\{ 1,\ldots, N\}$ as follows.
The encoder $\varphi$ is given as 
\begin{align*}
|\varphi(i)\rangle\langle\varphi(i)| 
= 
|0\rangle\langle 0|^{\otimes (i-1)}
\otimes 
|\alpha\rangle\langle \alpha|)
\otimes 
|0\rangle\langle 0|^{\otimes (n-i)},
\end{align*}
where $|\alpha|^2=E$.
The decoder $\{Y_i\}$ is as follows.
\begin{align*}
Y_i:= 
|0\rangle\langle 0|^{\otimes (i-1)}
\otimes 
(I-|0\rangle\langle 0|)
\otimes 
|0\rangle\langle 0|^{\otimes (n-i)}.
\end{align*}
Then, the error probability is
\begin{align*}
1- \langle\varphi(i)| Y_i |\varphi(i)\rangle
=1- e^{-E}.
\end{align*}
The above code is realizable with current technology.
This construction suggests that increasing the number of pulses yields quantum advantages over a classical Gaussian channel.
However, it is not so easy to increase the number of pulses for a fixed period in optical communication.
Development of a physical scheme to increase the number of pulses is required.
Exploring such a scheme remains a topic for future study.

\section*{Acknowledgement}
This research
was partially supported by a Grant-in-Aid for Scientific Research in the Priority Area `Deepening and Expansion of Statistical Mechanical Informatics (DEX-SMI)', No. 18079014
and a MEXT Grant-in-Aid for Young Scientists (A) No. 20686026.
The author thanks Professor Satoshi Ishizaka for interesting discussions.

\appendix

First, we prove the inequality 
\begin{align}
I(\epsilon| \bm{P},\bm{W})\ge
\overline{H}_+(\epsilon| \bm{W_P}).
\Label{1-30-9}
\end{align}
For this purpose,
we apply the discussion of Theorem 1 in Nagaoka-Hayashi\cite{N-H}
to the hypothesis testing 
$W_x^{(N)} \times P^{(N)}(dx)$ vs $W_{P^{(N)}}^{(N)} \times P^{(N)}(dx)$ 
on the composite system between 
the quantum system ${\cal H}^{(N)}$ and the classical system ${\cal X}^{(N)}$.
In this case, the sequence of sets of projections 
$\{ 
\{W_x^{(N)} -e^{\alpha_N a} W_{P^{(N)}}^{(N)} \le 0 \}
 \}_{x \in {\cal X}^{(N)}} \}_{N}$
yields the best test.
Choose an arbitrary real number $b_0$ satisfying
\begin{align}
b_0 >I(\epsilon| \bm{P},\bm{W}).
\Label{1-30-6}
\end{align}
Thus, when any sequence of sets of projections $\{ \{A_{x}^{(N)} \}_{x \in {\cal X}^{(N)}} \}_N$
satisfies the condition 
\begin{align}
\liminf_{N\to \infty}
\frac{-1}{a_N} \log \int_{{\cal X}^{(N)}}
\Tr W_{P^{(N)}}^{(N)} A_{x}^{(N)} P^{(N)}(dx)
\ge b_0,\Label{1-30-4}
\end{align}
then
\begin{align*}
\limsup_{N\to \infty}
 \int_{{\cal X}^{(N)}}
\Tr W_{x}^{(N)} (I-A_{x}^{(N)})
P^{(N)}(dx) > \epsilon.
\end{align*}

Define the projections
\begin{align*}
B_{b,N}:=
\{
I - e^{a_N b} W_{P^{(N)}}^{(N)} > 0 \}
\end{align*}
and 
\begin{align*}
B_{b,N,x}:=
\left\{
\begin{array}{ll}
\frac{1}{\Tr B_{b,N}W_{x}^{(N)} }B_{b,N}W_{x}^{(N)}B_{b,N}
 & \hbox{ if }B_{b,N} \neq 0 \\
0 & \hbox{ if }B_{b,N}=0.
\end{array}
\right.
\end{align*}
Then, 
the quantity $\overline{H}_+(\epsilon| \bm{W_P})$ can be expressed as follows.
\begin{align*}
\overline{H}_+(\epsilon| \bm{W_P})
&:= \sup_{b}
\left\{b\left|
\limsup_{N \to\infty}
\Tr W_{P^{(N)}}^{(N)} (I-B_{b,N})
\le \epsilon\right.\right\} \\
&= \sup_{b}
\left\{b \left|
\limsup_{N \to\infty}
\int_{{\cal X}^{(N)}}
\Tr [W_x^{(N)}(I-B_{b,N}) ]P(dx) 
\le \epsilon \right.\right\}\\
&= \sup_{b}
\left\{b \left|
\limsup_{N \to\infty}
\int_{{\cal X}^{(N)}}
\Tr [W_x^{(N)}(I-B_{b,N,x}) ]P(dx) 
\le \epsilon\right.\right\}.
\end{align*}
Since $B_{b,N,x} \subset B_{b,N}$,
all eigenvalues of 
$B_{b,N,x} W_{P^{(N)}}^{(N)}B_{b,N,x}$
are less than $e^{-a_N b}$.
Since 
$B_{b,N,x}$ is a rank-one projection or zero matrix,
$B_{b,N,x} W_{P^{(N)}}^{(N)}B_{b,N,x}
\le
e^{-a_N b} W_x^{(N)}$, which implies that
\begin{align*}
\Tr B_{b,N,x} W_{P^{(N)}}^{(N)}=
\Tr B_{b,N,x} W_{P^{(N)}}^{(N)}B_{b,N,x}
\le e^{-a_N b}.
\end{align*}
Thus,
$B_{b_0,N,x}$ satisfies the condition (\ref{1-30-4}).
Therefore,
\begin{align*}
\limsup_{N \to\infty}
\int_{{\cal X}^{(N)}}
\Tr [W_x^{(N)}(I-B_{b_0,N,x}) ]P(dx) 
> \epsilon,
\end{align*}
which implies that
$b_0 \ge \overline{H}_+(\epsilon| \bm{W_P})$.
Since $b_0$ is an arbitrary real number satisfying (\ref{1-30-6}),
the relation (\ref{1-30-9}) holds.

Next, we prove the opposite inequality
\begin{align}
 I(\epsilon| \bm{P},\bm{W})\ge
\overline{H}_+(\epsilon| \bm{W_P}).
\Label{1-30-7}
\end{align}

Let $\delta$ be an arbitrary real number satisfying that $\delta>0$.
Define the vector $|\phi_{N,x}\rangle$ by
\begin{align*}
|\phi_{N,x}\rangle \langle \phi_{N,x}|
=
\left\{
\begin{array}{ll}
\{W_x^{(N)} -e^{a_N (b+\delta)} W_{P^{(N)}}^{(N)} \ge 0 \}
& \hbox{ if }\{W_x^{(N)} -e^{a_N (b+\delta)} W_{P^{(N)}}^{(N)} \le 0 \} \neq 0\\
W_x^{(N)} & 
\hbox{ if }
\{W_x^{(N)} -e^{a_N (b+\delta)} W_{P^{(N)}}^{(N)} \le 0 \} = 0.
\end{array}
\right.
\end{align*}
Then, 
\begin{align*}
\langle \phi_{N,x}|
e^{a_N (b+\delta)} W_{P^{(N)}}^{(N)}
|\phi_{N,x}\rangle 
\le 1.
\end{align*}
That is,
\begin{align*}
\langle \phi_{N,x}|
e^{a_N b} W_{P^{(N)}}^{(N)}
|\phi_{N,x}\rangle 
\le e^{-a_N \delta}.
\end{align*}
Thus, the relation $
\{ e^{a_N b} W_{P^{(N)}}^{(N)} -I \ge 0 \}
\le e^{a_N b} W_{P^{(N)}}^{(N)}$ yields the result that
\begin{align*}
\langle \phi_{N,x}|
(I-B_{b,N})
|\phi_{N,x}\rangle 
=
\langle \phi_{N,x}|
\{ e^{a_N b} W_{P^{(N)}}^{(N)} -I \ge 0 \}
|\phi_{N,x}\rangle 
\le e^{-a_N \delta}.
\end{align*}
Since $B_{b,N}$ is a projection,
\begin{align*}
\|
|\phi_{N,x}\rangle \langle \phi_{N,x}|
-
B_{b,N}|\phi_{N,x}\rangle \langle \phi_{N,x}|B_{b,N}
\|_1
\le 2\sqrt{
 \langle \phi_{N,x}|
(I-B_{b,N})
|\phi_{N,x}\rangle}
\le 2e^{-\frac{a_N \delta}{2}}.
\end{align*}
Thus,
\begin{align*}
& 1- \Tr W_x^{(N)} 
\{W_x^{(N)} -e^{a_N (b+\delta)}W_{P^{(N)}}^{(N)} \le 0 \}
=
\Tr W_x^{(N)} 
|\phi_{N,x}\rangle\langle \phi_{N,x}|
\le
\Tr W_x^{(N)} 
B_{b,N}|\phi_{N,x}\rangle \langle \phi_{N,x}|B_{b,N}
+ e^{-a_N \delta} \\
\le &
\Tr W_x^{(N)} 
B_{b,N}
+ e^{-a_N \delta}
=
1- \Tr W_x^{(N)} (I-B_{b,N})
+ e^{-a_N \delta}.
\end{align*}
Therefore,
\begin{align*}
\limsup_{N \to \infty}
\int_{{\cal X}^{(N)}}
\Tr W_x^{(N)} 
\{W_x^{(N)} -e^{a_N (b+\delta)} W_{P^{(N)}}^{(N)} \le 0 \}
P^{(N)}(dx) 
\ge
\limsup_{N \to \infty}
\int_{{\cal X}^{(N)}}
\Tr W_x^{(N)} 
(I-B_{b,N})
P^{(N)}(dx),
\end{align*}
which implies that
$ I(\epsilon| \bm{P},\bm{W}) -\delta\le \overline{H}_+(\epsilon| \bm{W_P})$.
Since $\delta$ is an arbitrary positive real number,
we obtain (\ref{1-30-7}).

\end{document}